\documentclass[11pt]{article}
\usepackage{cite}
\usepackage{amsmath,amssymb,amsfonts}
\usepackage{amsthm}
\usepackage{algorithmic}
\usepackage{graphicx}
\usepackage{textcomp}
\def\BibTeX{{\rm B\kern-.05em{\sc i\kern-.025em b}\kern-.08em
    T\kern-.1667em\lower.7ex\hbox{E}\kern-.125emX}}

\usepackage{tikz}
\usetikzlibrary{calc,patterns,angles,quotes}
\usepackage{cases}
\usepackage{float}
\usepackage{subcaption}
\usepackage{xcolor}

\definecolor{color_selva}{RGB}{34,139,34}
\definecolor{lightgray}{gray}{0.9}

\newtheorem{thm}{Theorem}
\newtheorem{coro}{Corollary}

\DeclareMathOperator{\sgn}{sgn}

\providecommand{\keywords}[1]
{
  \small	
  \textbf{\textit{Keywords---}} #1
}

\begin{document}
\title{A Note on the Pure Nash Equilibria for Evolutionary Games on Networks}
\author{Jean Carlo Moraes 
\thanks{The Author is with the Federal University of Rio Grande do Sul, Porto Alegre, RS Brazil (e-mail: jean.moraes@ufrgs.br). }
\thanks{This work was supported in part by the CAPES/Brazil under Grant CAPES-PRINT7608979, project 88881.309849/2018-1. }
}

\maketitle
\renewcommand{\figurename}{Fig.}
\begin{abstract}
Recently, a new model extending the standard replicator equation to a finite set of players connected on an arbitrary graph was developed in evolutionary game dynamics. The players are interpreted as subpopulations of multipopulations dynamical game and represented as vertices of the graph, and an edge constitutes the relation among the subpopulations. At each instant, members of connected vertices of the graph play a 2-player game and collect a payoff that determines if the chosen strategies will vanish or flourish. The model describes the game dynamics of a finite set of players connected by a graph emulating the replicator dynamics. It was proved a relation between the stability of the mixed equilibrium with the topology of the network. More specifically, the eigenvalues of the Jacobian matrix of the system evaluated at the mixed steady state are the eigenvalues of the graph's adjacency matrix multiplied by a scalar. This paper studies the pure (strict) Nash equilibria of these games and how it connects to the network. We present necessary and sufficient conditions for a pure steady-state in coordination or anti-coordination game to be a (strict) Nash Equilibrium.
\end{abstract}

\keywords{Complex networks, evolutionay game theory, pure nash equilibrium, replicator equation on networked populations}

\section{Introduction}
\label{sec:introduction}
In order to explain the existence of persistent behaviours that animals show in conflict situations, the evolutionary game theory emerged in the early 70s. Until then, it was believed that this type of behaviour emerged because it was beneficial to the species, but this idea goes against the Darwinian thought, where selection occurs at the individual level, not at the community level. John Maynard Smith and George Price proposed a solution to this problem by introducing a type of game that does not require players to be rational, \cite{SmPr73}. Evolutionary Game Theory only requires individuals to have a strategy and nature to have a way to measure the fitness of this strategy against all others. This strategy is tested through a game, and the payment is its perpetuation (throughout reproduction, descendants will inherit your strategy). The success of a strategy depends on its performance against opposing strategies and on how often these strategies are found in the population. In the decade that followed the publication of the work of Price and Smith, there was a continuous growth in the applications of evolutionary games within biology, \cite{smith_1982}. 

During this same period, Taylor and Jonker modeled the dynamics of the evolutionary games in the continuous case, with a system of non-linear differential equations, called replicator equations, and in the discrete case, with a system of non-linear difference equations, \cite{TAJO78}. The replication equation was used to describe not only biological phenomena, such as mutation and viral spread \cite{weibull_1985} but also socioeconomic dynamics, as cooperation  \cite{weibull_1985} and dissemination of knowledge and wealth in the training of social networks \cite{MaVi07},\cite{JaZe15}. In 2015, Madeo and Mocenni introduced a model where they integrate the multipopulational replicator equation  \cite{weibull_1985} with finite play structure over a graph \cite{OhNo06} \cite{JaZe15}, \cite{MaWeFu14}. They proved three interesting results: the model presented is an extension of the replicator equation; for games with two strategies, the strategy profile is a (strict) Nash equilibrium of the game and if and only if is an (asymptotically) stable steady states of the system and, the most surprising, when all players have the same payoff matrices then the coordinates of the internal equilibrium is independent of the topology of the network, though the linear stability depends on the eigenvalues of the adjacency matrix of the network. 
They were able to connect the stability of the internal equilibrium with the network's topology, and therefore give conditions to have a strict Nash Equilibrium internal to the simplex based on the network. It was not presented any result of this kind made for the pure equilibria in general.

 In this paper, we present conditions on the network for coordination and anti-coordination games to (strict) Nash Equilibrium, even in the landscape where all players may have different payoff matrices. Our results provide many necessary (and sufficient) conditions for pure equilibria to be a (strict) Nash Equilibrium. It allows us, in many cases, to reduce considerably the number of cases to analyse. We will provide an example of a game with 262,144 pure equilibria. Using a combination of the results in the paper, we can conclude that only 16 can be a strict Nash equilibrium (8 are).

\section{Evolutionary Game on Networks (EGN)}
Consider a finite population of players $V=\{1,2,...,N\}$ connected in a network represented by an undirected graph $G$ with adjacency matrix $A=\{a_{v,w}\}$, where $a_{v,w}=1$ if player $v$ and $w$ are connected and  $a_{v,w}=0$ otherwise. The model developed in \cite{MaMo15} consider the interactions as games played by $N$ individuals as two-person games where an individual of a sub population $v$ plays against a represent of one of the sub populations $w$ connected to $v$, for simplicity we will say only that $v$ plays against $v$. We will consider a game with two pure strategies: cooperate (C) and defect (D). Let the strategy played by player $v$ be denoted by $x_v$, where $x_v = 1$ if $v$ cooperates and $x_v=0$ if $v$ defects. The payoff matrix for a player $v$ is given by 
 $$B_{v}=\begin{bmatrix}
b_{v,C,C} & b_{v,C,D} \\
b_{v,D,C} & b_{v,D,D}
\end{bmatrix}.$$
Player $v$ can also play a mixed strategy, in this case, $v$ cooperates with relative frequency $x_v$. Let ${\bf x}=(x_1,x_2,...,x_N)$ be the strategy profile in the game. Since $x_v \in \Delta=[0,1]$ for all $v \in V$, the strategy set $\mathcal{S}$ is $ \Delta^N=[0,1]^N$.
The payoff function of the two-person game player for player $v$ against a neighbour $w$ is given by
\begin{align*}
\phi(x_v,x_w)=&\begin{bmatrix}
x_v & (1-x_v) 
\end{bmatrix}B_v\begin{bmatrix}
x_w  \\
1-x_w
\end{bmatrix}=\\
=&(b_{v,C,C}-b_{v,D,C}+b_{v,D,D}-b_{v,C,D})x_v x_w -\\ -(b_{v,D,D}&-b_{v,C,D})x_v+b_{v,D,C}x_w+b_{v,D,D}(1-x_w)
\end{align*}
The payoff of player $v$ over the network is the sum of all the outcomes of the two-players games played by $v$ against his neighbours:
\begin{equation}
\phi_v({\bf x})=\sum_{w=1}^{N} a_{v,w} \phi(x_v,x_w).
\end{equation}
In the $(N,2)$-game case ($N$ players, $2$ strategies), the replicator equation on graphs, which models the dynamics of an evolutionary game on networked (EGN) populations, as defined in \cite{MaMo15} by the system of $N$ differential equation given by
\begin{equation}
\dot{x}_v= x_y(1-x_v)f_v({\bf x}), \text{for} \, v \in V, \label{eq:prin_red}
\end{equation}
where $\displaystyle{f_v({\bf x})=\frac{\partial \phi_v}{\partial x_v}({\bf x}).}$
If we denote $\sigma_{v,C}=b_{v,C,C}-b_{v,D,C}$ and $\sigma_{v,D}=b_{v,D,D}-b_{v,C,D}$, then
\begin{align*}
f_v({\bf x})&=\sigma_{v,C}\sum_{w=1}^{N} a_{v,w}x_w - \sigma_{v,D}\sum_{w=1}^{N} a_{v,w}(1-x_w)=\\
&=(\sigma_{v,C}+\sigma_{v,D}) \sum_{w=1}^{N} a_{v,w}x_w - \sigma_{v,D}\sum_{w=1}^{N} a_{v,w}.
\end{align*}

Note that $\sum_{w=1}^{N} a_{v,w}$ is the degree of the vertex  $v$. Let us denote it by $d_v$, then equations  \eqref{eq:prin_red} can be rewritten as 
\begin{equation}
\dot{x}_v= x_y(1-x_v)\bigg[(\sigma_{v,C}+\sigma_{v,D}) \sum_{w=1}^{N} a_{v,w}x_w - \sigma_{v,D}d_v\bigg].
\end{equation}

Given $ v \in V$ and ${\bf x} \in \Delta^N$ we shall use the notation $(x_v,{\bf x_{-v}})$ in order to define Nash Equilibria. The vector ${\bf x_{-v}}$ is a vector in $\Delta^{N-1}$ which corresponds to the vector ${\bf x}$ without the component $v$, ${\bf x_{-v}}=(x_1,x_2,...,x_{v-1},x_{v+1},...,x_N)$. Therefore, for all $ v \in V$ we define $\pi_v:\Delta \times \Delta^{N-1} \rightarrow \mathbb{R}_{+}$ as
\begin{equation}
\pi_v(x_v,{\bf x_{-v}})=\phi_v({\bf x}).
\end{equation}

A Nash equilibrium NE is a profile of strategies where no player can improve its payoff unilaterally, i.e., for all players $v$, $\pi_v$ cannot be increased, changing only the strategy $x_v$. 
\begin{align*}
\Theta^{NE}=\{{\bf x} \in \Delta^N : \forall  v  \, \phi_v(x_v, {\bf x_{-v}}) \geq \phi_v(y, {\bf x_{-v}});\, \forall \, y \in \Delta\}.
\end{align*}
If for all players $v$ $x_v$ is the unique best response to others players' strategies in that equilibrium ${\bf x_{-v}}$ then ${\bf x}$ is a strict Nash equilibrium SNE.
\begin{align*}
\Theta^{SNE}=\{{\bf x} \in \Delta^N : \forall  v \;  \phi_v(x_v, {\bf x_{-v}}) > \phi_v(y, {\bf x_{-v}});\, \forall y \in \Delta\}.
\end{align*}

Following \cite{MaMo15}, for the EGN model given by \eqref{eq:prin_red} these sets can be rewritten as:

\begin{align*}
\Theta^{NE}=\bigg\{{\bf x} \in & \Delta^N : \forall \, v \,\, \Big( (x_v=0 \land f_v({\bf x}) \leq 0) \,\, \lor \\
&(x_v=1 \land f_v({\bf x}) \geq 0)\,\,  \lor \,\,(f_v({\bf x}) = 0) \Big) \bigg\}.
\end{align*}
and
\begin{align*}
\Theta^{SNE}=\bigg\{{\bf x} \in & \Delta^N : \forall \, v \,\, \Big( (x_v=0 \land f_v({\bf x}) < 0) \,\, \lor \\
&(x_v=1 \land f_v({\bf x}) > 0)\,\,  \lor \,\,(f_v({\bf x}) = 0) \Big) \bigg\}.
\end{align*}
The set of steady states, $\Theta^{*}$ of  the system of ODEs\eqref{eq:prin_red} is the set that contains all points in $\Delta^N$ such that $x_v(1-x_v)f_v({\bf x})=0$ for all $v$ in $V$, therefore:
\begin{align*}
\Theta^{*}=\bigg\{{\bf x} \in & \Delta^N : \forall \, v \,\,  \big(x_v=0 \, \lor x_v=1 \,  \lor \,f_v({\bf x}) = 0 \big) \bigg\}.
\end{align*}

Looking to the definition $\Theta^*$ it is clear that it contains the set $\Theta^p=\{0,1\}^N$. The elements in $\Theta^p$ are called pure steady states. Among all points in $\Theta^p$, we denote two special ones where we have full cooperation, $\mathbf{x_{FC}}=(1,1,...,1)$, and full defection, $\mathbf{x_{FD}}=(0,0,...,0)$. We can partition the steady states set in three  subsets: the pure $(\Theta^p)$, the mixed $(\Theta^m)$ and the pure/mixed $(\Theta^{mp})$. Thus, $\Theta^* = \Theta^p \cup \Theta^m \cup \Theta^{mp}$. The set $\Theta^m$ contains all steady states in the interior of the simplex $\Delta^N$, so $\Theta^m= (0,1)^N \cap \Theta^*$. Note that if ${\bf x} \in \Theta^m$ then ${\bf x} \in \Theta^*$ which implies that $f_v({\bf x})=0$ for all $v$ and therefore ${\bf x}$ is a Nash equilibrium, $\Theta^m \subset \Theta^{NE}$. All the other steady states are classified as pure/mixed.

One of the main results in \cite{MaMo15} establish that if $\sigma_{v,C}$ and $\sigma_{v,D}$ are constant and not zero for all $v$ then there exists $\mathbf{x^*} \in \Theta^m$. Moreover, if the adjacency matrix $A$ of the graph is invertible, then the values of  ${\bf x^*}$ do not depend on $A$, but the linear stability depends on the eigenvalues of $A$. In this case, the topology plays no role in the location of the steady-state, but it does on its stability. Other result in \cite{MaMo15} establishes conditions for a pure steady-state to be in $\Theta^{NE}$ (and also in $\Theta^{NES})$. These conditions depend on the mathematical expression that derives from equation \eqref{eq:prin_red}. Unlike the result for mixed steady state, in \cite{MaMo15} is not provided conditions on the network that ensures that a point in $\Theta^p$ is a NE or SNE. Results for two special pure equilibrium, $\mathbf{x_{FC}}$ and $\mathbf{x_{FD}}$, are obtained in \cite{MaMo21}.  Such a connection for a broader theory may not be possible since the size of $\Theta^p$ can be very large and increases exponentially when the number of vertices increases. However, in this paper, we provide results for this question depending on the ratio of neighbours that plays cooperate and defect, which depends on the network's topology. This result allows us, in many situations, to reduce the number of possible (strict) Nash equilibria. For some network structures, the results allow us to find among the elements $\Theta^p$ all the pure (strict) Nash equilibrium.

\section{Results for Pure Nash Equilibria}

In order to analyse the pure steady states of \eqref{eq:prin_red} we shall need the Jacobian matrix associated with it,

\begin{equation}
J_{v,w}(\mathbf{x^*}) = \begin{cases} (1-2x^*_{v})k_v({\mathbf{x^*}}) &\mbox{if } w=v \\ 
x^*_v(1-x^*_v)(\sigma_{v,C}+\sigma_{v,D})a_{v,w} & \mbox{if } w \neq v   \end{cases} \label{Jaco}
\end{equation}
where 
\begin{equation}
k_v({\mathbf{x^*}})=\sum_{w=1}^{N} a_{v,w} \bigg( \sigma_{v,C}x^*_w - \sigma_{v,D}(1-x^*_w) \bigg).
\end{equation}

For ${\mathbf x^*} \in \Theta^p$, we define
\begin{equation}
N_{v,C}:=\sum_{u=1}^{N} a_{v,u}x^*_u \quad \text{and} \quad
N_{v,D}:=\sum_{u=1}^{N} a_{v,u}(1-x^*_u)
\end{equation}

Note that $N_{v,C}$ is the amount of neighbours of $v$ that cooperates and $N_{v,D}$ is the amount of neighbours of $v$ that defects. Also important to note that $N_{v,C}+N_{v,D}=d_v$. With these definitions we can state our main result which will connect the linear stability of a pure steady state (and the Nash equilibrium) with the ratio $\displaystyle{R_v:=\sigma_{v,C}(\sigma_{v,D}})^{-1}$ and the quantities $N_{v,C}$ and $N_{v,D}$.

\begin{thm}\label{theom}
Let $\mathbf{x^*} \in \Theta^p$, $\mathbf{x^*}$ is a steady state of the EGN defined \eqref{eq:prin_red}. Then $\mathbf{x^*}$ is a NE iff for each vertex $v$
\begin{itemize}
\item that has a coordination payoff matrix, we have
\begin{numcases}{}
$$N_{v,D} \leq R_v N_{v,C}$$ & if $x^*_v=1$,  $v$  plays cooperate;
 \label{BI_CO_NA}  \\ 
$$N_{v,C} \leq \frac{1}{R_v}N_{v,D}$$ & if   $x^*_v=0$, $v$ plays defect; \label{BI_DE_NA}
\end{numcases}
     
\item that has a anti-coordination payoff matrix, we have that
\begin{numcases}{}
$$N_{v,C} \leq \frac{1}{R_v}N_{v,D}$$  & if $x^*_v=1$, $v$  plays cooperate; \label{CO_CO_NA}  \\ 
$$N_{v,D} \leq R_v N_{v,C}$$ & if $x^*_v=0$, $v$ plays defect. \label{CO_DE_NA} 
\end{numcases}
\end{itemize} 

Moreover, $\mathbf{x^*}$ is a SNE iff for each vertex $v$
\begin{itemize}
\item that has a coordination payoff matrix, we have
\begin{numcases}{}
$$N_{v,D} < R_v N_{v,C}$$ & if $x^*_v=1$,  $v$  plays cooperate;
 \label{BI_CO}  \\ 
$$N_{v,C} < \frac{1}{R_v}N_{v,D}$$ & if   $x^*_v=0$, $v$ plays defect; \label{BI_DE}
\end{numcases}
     
\item that has a anti-coordination payoff matrix, we have that
\begin{numcases}{}
$$N_{v,C} < \frac{1}{R_v}N_{v,D}$$  & if $x^*_v=1$, $v$  plays cooperate; \label{CO_CO}  \\ 
$$N_{v,D} < R_v N_{v,C}$$ & if $x^*_v=0$, $v$ plays defect. \label{CO_DE} 
\end{numcases}
\end{itemize} 

\end{thm}

\begin{proof}
Given $\mathbf{x^*} \in \Theta^p$, for all $v$ $x_v^*=0$ or $x_v^*=1$, which implies, from \eqref{Jaco} that $J_{v,w}(\mathbf{x^*})=0$ if $v \neq w$. Thus, the Jacobian is a diagonal matrix. The elements of diagonal are its eigenvalues and are given by 
\begin{align*}
\lambda_v &= J_{v,v}(\mathbf{x^*}) = (1-2y^*_{v})(\sigma_{v,C} k_{v,C}(\mathbf{x^*}) - \sigma_{v,D} k_{v,D}(\mathbf{x^*}))=\\
&= (1-2x^*_{v})\bigg(\sigma_{v,C} \sum_{u=1}^{N} a_{v,u}x^*_u - \sigma_{v,D}\sum_{u=1}^{N} a_{v,u}(1-x^*_u)\bigg)
\end{align*}

By Theorem 2 in \cite{MaMo15}, $\mathbf{x^*}$ is a SNE (NE) if and only if the eigenvalues of the Jacobian matrix $J(\mathbf{x^*})$ are all less (less or equal) than zero.
\begin{align*}
\lambda_v &= J_{v,v}(\mathbf{x^*}) = (1-2x^*_{v})\bigg(\sigma_{v,C}N_{v,C} - \sigma_{v,D}N_{v,D} \bigg)=\\
&= \begin{cases} \sigma_{v,C}N_{v,C} - \sigma_{v,D}N_{v,D}  &\mbox{if } x^*_{v}=0 \\ 
\sigma_{v,D}N_{v,D}- \sigma_{v,C}N_{v,C}  & \mbox{if }  x^*_{v}=1   \end{cases}
\end{align*}

Let $\displaystyle{R_v:=\frac{\sigma_{v,C}}{\sigma_{v,D}}}$, then we can write
\begin{align*}
\sigma_{v,D}\lambda_v = \begin{cases} R_v N_{v,C} - N_{v,D}  &\mbox{if } x^*_{v}=0 \\ 
N_{v,D}- R_vN_{v,C}  & \mbox{if }  x^*_{v}=1,   \end{cases}
\end{align*}
in a short way,
$$ \sigma_{v,D}\lambda_v=(2x_v^*-1)\big(N_{v,D}- R_vN_{v,C}\big).$$
Therefore,
\begin{align*}
\sgn(\lambda_v)&=\sgn(\sigma_{v,D}) \sgn(\sigma_{v,D}\lambda_v)= \nonumber \\ 
&=\sgn(\sigma_{v,D})\sgn(2x_v^*-1)\sgn(N_{v,D}- R_vN_{v,C})
\end{align*}

If $\sgn(\sigma_{v,D})=\sgn(\sigma_{v,C})>0$, then
\begin{align*}
\lambda_v < (\leq) \,\, 0 & \iff
\sgn(2x^*-1)\neq \sgn(N_{v,D}- R_vN_{v,C}) \\
& \iff \begin{cases} N_{v,D} <(\leq)\,\, R_v N_{v,C}, & \mbox{if  } x^*_v=1;    \\ 
 N_{v,C} < (\leq)\,\, \frac{1}{R_v}N_{v,D}, & \mbox{if  } x^*_v=0; \end{cases}
\end{align*} 

In the other hand, if $\sgn(\sigma_{v,D})=\sgn(\sigma_{v,C})<0$, then
\begin{align*}
\lambda_v < (\leq) \,\, 0 & \iff
\sgn(2x^*-1)=\sgn(N_{v,D}- R_vN_{v,C}) \\
& \iff \begin{cases}  N_{v,C} < (\leq) \,\, \frac{1}{R_v}N_{v,D},  & \mbox{if  } x^*_v=1;   \\ 
N_{v,D} < (\leq)\,\, R_v N_{v,C}, & \mbox{if  } x^*_v=0;     \end{cases}
\end{align*} 
\end{proof}


\begin{coro}\label{Coro1} Let $\mathbf{x^*_{FC}}=(1,1,\dots,1)$, $\mathbf{x^*_{FD}}=(0,0,\dots,0)$. If for all vertices $v$, $\sgn(\sigma_{v,C})=\sgn(\sigma_{v,D})>0$, then $\mathbf{x^*_{FC}}$ and $\mathbf{x^*_{FD}}$ are SNE. In the other hand, if $\sgn(\sigma_{v,C})=\sgn(\sigma_{v,D})<0$ $\forall v$, then $\mathbf{x^*_{FC}}$ and $\mathbf{x^*_{FD}}$ are not NE.
\end{coro}
\begin{proof}
Since  all players cooperate in the steady state $\mathbf{x^*_{FC}}$, no player $v$ has a neighbour that defects, i.e. $N_{v,D}=0 \,\, \forall \, v$. Therefore, \eqref{BI_CO} is always satisfied and \eqref{CO_CO_NA} is never satisfied. Analogous,  $N_{v,C}=0 \, \forall \, v$ in the steady state $\mathbf{x^*_{FD}}$, thus  condition~\eqref{BI_DE} is always satisfied and \eqref{CO_DE_NA} is never satisfied.  Thus,  $\mathbf{x^*_{FC}}$ and $\mathbf{x^*_{FD}}$ are SNE in coordination games and it is not a NE in anti-coordination games.
\end{proof}

\begin{coro}\label{coro2}
Let $\mathbf{x^*} \in \Theta^p$, $\mathbf{x^*}$ is a steady state of the EGN given by \eqref{eq:prin_red}. If there is a vertex $v$, with  $\sgn(\sigma_{v,C})=\sgn(\sigma_{v,D})>0$, such that all his neighbours plays a strategy different than him, then $\mathbf{x^*}$ is not a NE.
\end{coro}
\begin{proof}
Let $\mathbf{x^*}\in \Theta^p$. Assume $x^*_k=1$ for some $k \in \{1,2,...,N\}$ and $x^*_l=0$ for all $l \neq k$, then $N_{k,C}=0, \, N_{k,D}=d_k$ and \eqref{BI_CO_NA} is not satisfied. Analogous, if $x^*_k=0$ for some $k \in \{1,2,...,N\}$ and $x^*_l=1$ for all $l \neq k$, then $N_{k,D}=0, \, N_{k,C}=d_k$ and condition~\eqref{BI_DE} is not satisfied.
\end{proof}

An example of the corollary above application is to graphs with leaves (one vertex link only one other in the graph). If the leaf has a coordination payoff matrix, it must play the same strategy to its neighbour in a pure strategy profile. For instance, consider a star graph, where all players have a coordination payoff matrix, then a $\mathbf{x^*_{FC}}$ and $\mathbf{x^*_{FD}}$ are the only pure NE of the game.

\begin{coro}\label{Coro3}
Let $\mathbf{x^*} \in \Theta^p$, $\mathbf{x^*}$ is a steady state of the EGN. If there is a vertex $v$, such that $\sigma_{v,C}$ and $ \sigma_{v,D}$ are negatives and all his neighbours play the same strategy than him, then $\mathbf{x^*}$ is not a Nash equilibria.
\end{coro}
\begin{proof}
Let $\sgn(\sigma_{v,C})=\sgn(\sigma_{v,D})<0$, $\forall v$ and $\mathbf{x^*} \in Theta^p$. Assume for some $k \in \{1,2,...,N\}$ all neighbours of $k$ plays the same strategy than him. If  $x^*_k=0$, then $N_{k,D}=d_k$ and $N_{k,C}=0$, which implies that \eqref{CO_DE_NA} is not satisfied.  Analogous, if $y^*_k=1$, but in this case \eqref{CO_CO_NA} will be violated.
\end{proof}

For the next two corollaries we will assume that all players have the same payoff matrix and they are either coordination or anti-coordination type, i.e., $B_v=B$ for all $v \in V$ and $\sgn(\sigma_{C})=\sgn(\sigma_{D})\neq 0$, where $\sigma_{C}=b_{11}-b_{12}$ and $\sigma_{D}=b_{22}-b_{21}$. Thus, $R_v=\frac{\sigma_{C}}{\sigma_{D}}=R$ for all $v$. A similar assumption is made in \cite{MaMo15} to show that the stability mixed equilibrium of the EGN depends only on the adjacency matrix of the graph and also in \cite{MaMo21} to obtain conditions on self-regulation parameters in order to be able to determine if the mixed equilibrium, $\mathbf{x^*_{fc}}$ and  $\mathbf{x^*_{fd}}$ are NE for Stag-Hunt Games and Chicken Games. The network structure may have an important role in the stability of many other pure equilibria of the EGN. A game with $N$ players has $2^N$ equilibria in $\Theta^p$. Their analysis is important to know if subpopulations of cooperators can emerge (and be stable) in the dynamics of the evolution game.


\begin{coro}\label{Coro4}
Suppose $B_v=B$ for all $v \in V$ in a coordination game. If $R_v=R>d_v-1$ for all $v \in V$,  then $\mathbf{x^*_{FC}}$ and $\mathbf{x^*_{FD}}$ are the unique SNE in $\Theta^p$.
\end{coro}
\begin{proof}
Let $\mathbf{x^*} \in \Theta^p$ such that $\mathbf{x^*}\neq \mathbf{x^*_{FC}}$ and $\mathbf{x^*}\neq \mathbf{x^*_{FD}}$. Then $\mathbf{x^*}$ has at least one player $k$ that defects and plays against a cooperator. Note that $\mathbf{x^*}$ has at least one player that plays $1$ and one that plays $0$, otherwise $\mathbf{x^*}$ would be $\mathbf{x^*_{FC}}$ or $\mathbf{x^*_{FD}}$. Moreover, at least one player $k$ that defects and plays against a cooperator, if that does not happen, then we would have a disconnected graph with all cooperators on one side and all defectors on the other. Since the graph is connected, this does not happen. Then $N_{k,C} \geq 1$ and
\begin{align*}
R > d_k-1&=N_{k,C}+N_{k,D}-1 \geq N_{k,D} \iff \\
&\iff R N_{k,C} > N_{k,D}.
\end{align*}

Thus, condition \eqref{BI_DE} is not satisfied.
\end{proof}
\subsection{Independent Dominating Sets}
A set $D \subset V$ is a dominating set of a graph $G$ with vertices $V$ if every vertex in $\in V-D$ is adjacent to a vertex in $D$.  A set  $D \subset V$ is independent if no two vertices in $D$ are adjacent. An independent dominating set of a graph $G$ is a set of vertices that are dominating and independent in $G$, \cite{Berge1961},\cite{GODDARD2013839}, \cite{Ore1961}.
The following corollary gives conditions on $R_v$'s such that a pure steady-state $\mathbf{x}^*$ of the EGN is a SNE in an anti-coordination if and only if the set of players that cooperates or the set of players that defects are an independent dominating set of a graph $G$.
\begin{coro}\label{Coro5}
Consider an EGN over a graph $G$ with all players having a payoff matrix anti-coordination $B_v$. Let $\mathbf{x^*} \in \Theta^p$, $\mathcal{D}:=\{v \, : x^*_v=0 \}$ and $\mathcal{C}:=\{v \, : x^*_v=1 \}.$ Then
\begin{itemize}
\item[i)] if $R_v >d_v-1$ for all $v \in V$, then $\mathbf{x^*}$ is a SNE iff $\mathcal{C}$ is an independent dominant subgraph of $G$.
\item[ii)] if $R_v^{-1} > d_v-1$ for all $v \in V$, then $\mathbf{x^*}$ is SNE iff $\mathcal{D}$ is an independent dominant subgraph of $G$.
\end{itemize}
\end{coro}
\begin{proof}
We will start showing \emph{i)}. Assume $R_v>d_v-1$ for all $v \in V$ and $\mathbf{x^*} \in \Theta^p$ is a SNE. First note that, anti-coordination implies by Corollary~\ref{Coro1} that $\mathcal{C} \neq \emptyset$. Now, by Corollary~\ref{Coro4}, a player that defects has at least one neighbour that cooperates, so  every vertex in $\in V-\mathcal{C}$ is adjacent to a vertex in $\mathcal{C}$, thus, $\mathcal{C}$ is dominant. In order to show that $\mathcal{D}$ is independent we shall prove that for every $v \in \mathcal{C}$, $N_{v,C}=0$, i.e., no neighbour of $v$ cooperates. Take $v \in \mathcal{C}$, if $d_v=1$ then Corollary~\ref{Coro3} implies that $N_{V,D}\geq1$ and therefore $N_{v,D}=1$ and $N_{v,C}=0$. If $d_v>1$, then $R_v>0$, since $R_v>d_v-1$. Thus $R_v^{-1}<(d_v-1)^{-1}$. Since $\mathbf{x^*} \in \Theta^p$ is a SNE, then equation \eqref{CO_CO} implies that $N_{v,C} < R_v^{-1}N_{v,D}$, thus
\begin{align*}
N_{v,C} < R_v^{-1}N_{v,D}&<\frac{N_{v,D}}{d_v-1} \Rightarrow N_{v,C}(d_v-1) < N_{v,D}\Rightarrow\\
&\Rightarrow d_v N_{v,C} < N_{v,D}+N_{v,D}=d_v\Rightarrow\\
&\Rightarrow d_v N_{v,C} < d_v\Rightarrow  N_{v,C} < 1 \Rightarrow  \\
&\Rightarrow  N_{v,C} =0.
\end{align*}
Therefore $\mathcal{C}$ is independent dominant subgraph of $\mathcal{G}$. 

In order to show the the reverse implication assume $\mathcal{C}$ is an independent dominant set of $G$, then for all $v \in \mathcal{C}$, $N_{v,C}=0$. Also, again Corollary~\ref{Coro3} implies that $N_{V,D}\geq1$, therefore condition \eqref{CO_CO} is easily checked for $v \in \mathcal{C}$ ($x_v=1$). Now, if $x^*_v=0$, then Corollary~\ref{Coro3} implies that $N_{V,C}\geq1$ which implies $N_{v,D}<d_v-1$. Thus, $R_v N_{v,C} > d_v-1 > N_{v,D}$ which is the condition \eqref{CO_DE}. Both conditions together ensures that $\mathbf{x^*}$ is a SNE.
The proof of \emph{ii)} is analogous.
\end{proof}

\emph{Remark:} In the case that $R_v=R$ for all $v \in V$, if $R>\max_{v \in V} (d_v-1)$, then $\mathbf{x^*}$ is a SNE  iff $\mathcal{C}$ is an independent dominant set of $G$. Analogous, if $R^{-1}<\max_{v \in V} (d_v-1)$, then $\mathbf{x^*}$ is an SNE iff $\mathcal{D}$ is an independent dominant set of $G$.

\section{Application} \label{sec:app}

In this section, we will provide some examples to show how helpful can be the Theorem~\ref{theom} and the corollaries of the previous section. As argued before, the graph's topology can be an essential instrument to determine the SNE (NE) in $\Theta^p$.

\subsection{Coordination Game and Caterpillar graph}
Consider a EGN over a Caterpillar graph $ G:=C_{8}(0,1,0,5,0,0,4,0)$. Each player in the main stalk (players $1$ to $8$) has payoff matrix $B_{cs}$ and each player in branches (players $9$ to $18$) has payoff matrix $B_b$, where 

$$B_{cs}=\begin{bmatrix}
2.1 & 0 \\
0 & 1
\end{bmatrix} \quad \text{and} \quad B_{b}=\begin{bmatrix}
3 & 0 \\
0 & 2
\end{bmatrix}.$$
\begin{figure}[!ht]
\centering
\includegraphics[width=7cm]{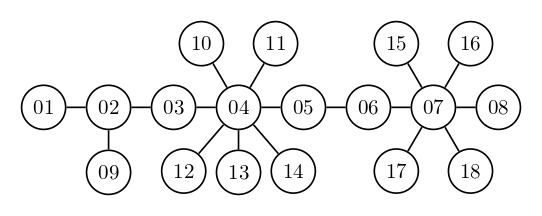}
\caption{Caterpillar graph $C_8(0,1,0,5,0,0,4,0)$}
\label{fig:cater}
\end{figure}

Let $R_{cs}:=R_v$ for $v \in\{1,2,...,8\}$ and $R_{b}:=R_v$ for  $v \in\{9,10,...,18\}$. Then $R_{cs}=2.1$ and $R_b=1.5$. Note that $\Theta^p$ contains $2^{18}=262,144$ pure equilibria. We will use the theory developed to show that most of them can not be SNE. Assume all players play a pure strategy, Corollary~\ref{coro2} guarantees that all leaves will play the same strategy that their neighbour in the main stalk. Therefore, players $15,16,17,18$ play in agreement with $7$, so does player $8$ which can be viewed as a leaf of $7$. Analogous, players $10,11,12,13,14$ will agree with player $4$. The same corollary will also state that in a SNE, players $1,2,9$ plays the same strategy. 

Player $3$ always agrees with player $2$. Let us verify that affirmation. Note that it has only two neighbours: $2$ and $4$. If $4$ defects then:
\begin{itemize}
\item $3$ cooperates if $2$ cooperates (since he will benefit more cooperating with $2$;
\item $3$ defects if $2$ defects since both its neighbours defects by Corollary~\ref{coro2}.
\end{itemize} 
If $4$ cooperates, then $3$ will cooperate, which implies that $2$ has to cooperate by Theorem~\ref{theom}. In any way that $4$ plays, $2$ and $3$ will always agree.
Last, players $5$ and $6$ will also play the same strategy. Suppose they could disagree, then one of them would play cooperate. The other one would have only two neighbours where at least one of them cooperates. Since the payoff for cooperating is better than defecting, he would be better at cooperating in agreement with his neighbour. Therefore we can group the players in 4 sets: $G_1:=\{1,2,3,9\}$, $G_2:=\{4,10,11,12,13,14\}$, $G_1:=\{5,6\}$, $G_1:=\{7,8,15,16,17,18\}$. 
\begin{figure}[!ht]
\centering
\includegraphics[width=7cm]{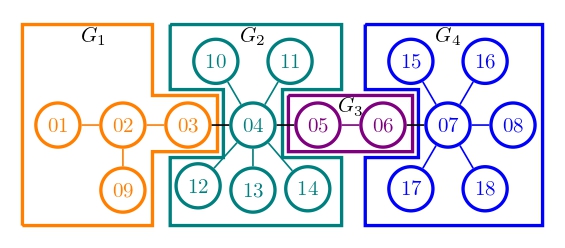}
\caption{Representation of the 4 sets which among them each player should play the same strategy.}
\end{figure}

An $x^* \in \Theta^p$ only is a SNE if the players in these groups play the same strategy. Thus, from $262,144$ equilibria in$\Theta^*$, only $2^4=16$ can be SNE. We use Theorem~\ref{theom} to verify that among these 16, 8 are, in fact, SNE. We present all of them in the figure below, where the players (groups) in yellow cooperates and those in red defects.

\begin{figure}[!ht]
\centering
\includegraphics[width=5.5cm]{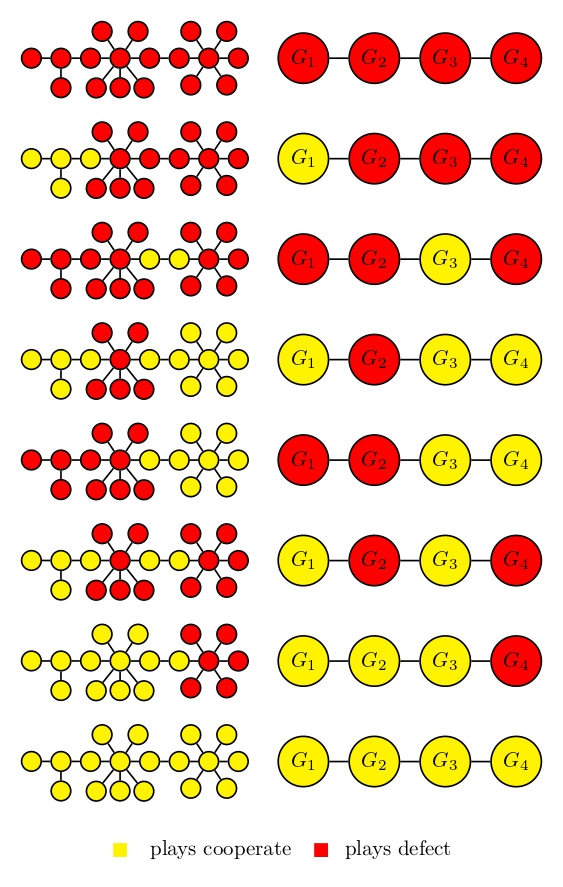}
\caption{Representation of 8 SNE of the problem. In the left column, it is the representation in the original graph; in the right column it is presented the reduced graph with the groups where all players inside it play in agreement.}
\end{figure}
\clearpage
\subsection{Anti-Coordination Game in an  Erd\"os-Rényi network}
Consider the game  over a Erd\"os-Rényi graph generated with $8$ vertices and average degree $4$ where all players have the same payoff matrix $B$, let $R=\sigma_{v,C}/\sigma_{v,D}$, for any player $v$. 
\begin{figure}[!ht]
\centering
\includegraphics[width=3cm]{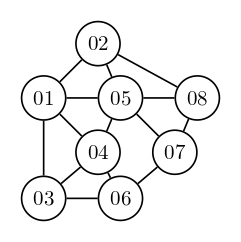}
\caption{Erd\"os-Rényi graph with $8$ vertices and average degree $4$.}
\label{fig:erdos}
\end{figure}

In the anti-coordination case, $\sgn(\sigma_C)=\sgn(\sigma_D)<0$, Corollary~\ref{Coro3} guarantees that if $\mathbf{x^*}$ is a SNE then every player has at least one neighbour that plays a different strategy different than him. Moreover, note that  player 5 is more connected to others than any other player, $d_5=5$. If $R>4$, then by Corollary~\ref{Coro5}, the SNE of this game are strategies profiles that the cooperators form an independent dominant set of the graph. There are $9$ of them, represented in the below.

\begin{figure}[!ht]
\centering
\includegraphics[width=6cm]{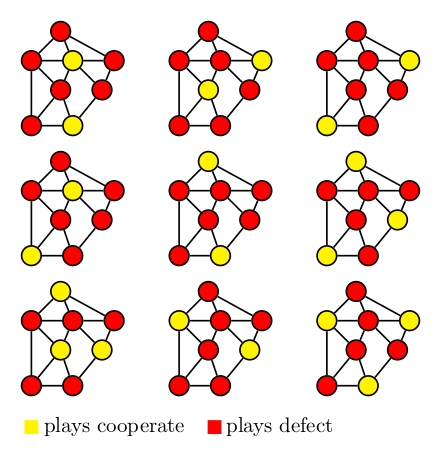}
\caption{The cooperator form an independent dominant graph for SNE for the EGN over the graph of Figure~\ref{fig:erdos} for $R>4$. }
\label{fig:id_graph}
\end{figure}

Suppose $R \in (0, \frac14)$, defectors would be an independent dominant set of the graph in Figure~\ref{fig:erdos}. Moreover, there would be $9$ solutions which would be the same as in Figure~\ref{fig:id_graph} where all cooperator becomes a defector and vice versa. For $R=2$, there are only 2 SNE among the 256 pure steady states.

\begin{figure}[!ht]
\centering
\includegraphics[width=6cm]{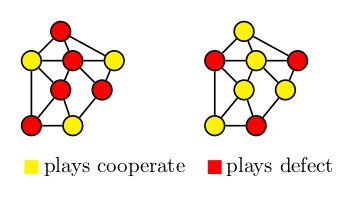}
\caption{Strict Nash equilibria for the EGN over the graph of Figure~\ref{fig:erdos} for $R=2$. }
\end{figure}

For $R$ ranging in $(0, \infty)$, it possible to check the conditions of Theorem~\ref{theom} for each one of the $256$ pure steady states. We only need to note that any change of states must occur in a ratio among degrees of a vertex, i.e., for a point in the set $\{\frac{d_v}{d_u} \,: \, v, u \in V \}$. If we order this set as $Ch:=\{0,\alpha_1, \alpha_2,..., \alpha_l\}$, then it is enough to test for $R \in (0,\alpha_1)$, any $R \in (\alpha_i, \alpha_{i+1})$ for $i=1,2,..l-1$ and $R=\alpha_i$ for $i=1,2,..l$. In our example, $Ch=\{0,\frac14,\frac13,\frac12,\frac23,1,\frac32,2,3,4\}$.  The diagram below shows the number of SNE for all values of $R \in (0, \infty)$ of the EGN over the graph of Figure~\ref{fig:erdos}. For clarity, the $R$-axis is not on scale.

\begin{figure}[!ht]
\centering
\includegraphics[width=9cm]{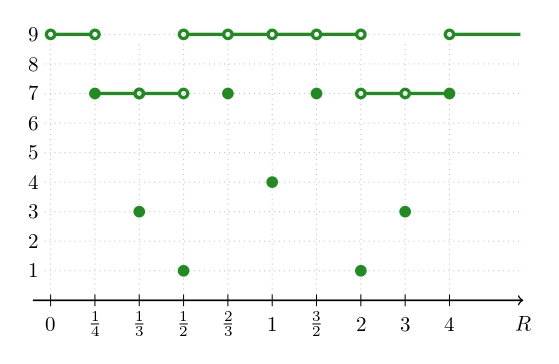}
\caption{Number of Pure Nash equilibria for $R>0$ of the EGN over the graph of Figure~\ref{fig:erdos}. }
\end{figure}

\clearpage
\bibliographystyle{ieeetr}
\bibliography{egn_bib}

\begin{thebibliography}{10}

\bibitem{SmPr73}
J.~Smith and G.~Price, ``The logic of animal conflict,'' {\em Nature},
  vol.~246, pp.~15--18, 1973.

\bibitem{smith_1982}
J.~M. Smith, {\em Evolution and the Theory of Games}.
\newblock Cambridge University Press, 1982.

\bibitem{TAJO78}
P.~D. Taylor and L.~B. Jonker, ``Evolutionary stable strategies and game
  dynamics,'' {\em Mathematical Biosciences}, vol.~40, no.~1, pp.~145--156,
  1978.

\bibitem{weibull_1985}
J.~Weibull, {\em Evolutionary Game Theory}.
\newblock Cambridge University Press, 1995.

\bibitem{MaVi07}
A.~DI~MARE and V.~LATORA, ``Opinion formation models based on game theory,''
  {\em International Journal of Modern Physics C}, vol.~18, no.~09,
  pp.~1377--1395, 2007.

\bibitem{JaZe15}
M.~O. Jackson and Y.~Zenou, ``Games on networks,'' in {\em Handbook of Game
  Theory with Economic Applications}, vol.~4, ch.~3, pp.~95--163, Elsevier,
  2015.

\bibitem{OhNo06}
O.~H. and N.~M., ``The replicator equation on graphs,'' {\em J Theor Biol.},
  vol.~243, no.~1, pp.~86--97, 2006.

\bibitem{MaWeFu14}
W.~Maciejewski, F.~Fu, and C.~Hauert, ``Evolutionary game dynamics in
  populations with heterogenous structures,'' {\em PLOS Computational Biology},
  vol.~10, pp.~1--16, 04 2014.

\bibitem{MaMo15}
D.~Madeo and C.~Mocenni, ``Game interactions and dynamics on networked
  populations,'' {\em IEEE Transactions on Automatic Control}, vol.~60, no.~7,
  pp.~1801--1810, 2015.

\bibitem{MaMo21}
D.~Madeo and C.~Mocenni, ``Consensus towards partially cooperative strategies
  in self-regulated evolutionary games on networks,'' {\em Games}, vol.~12,
  no.~3, 2021.

\bibitem{Berge1961}
C.~Berge, {\em The theory of graphs and its applications}.
\newblock London: Methuen \& Co. Ltd., 1962.

\bibitem{GODDARD2013839}
W.~Goddard and M.~A. Henning, ``Independent domination in graphs: A survey and
  recent results,'' {\em Discrete Mathematics}, vol.~313, no.~7, pp.~839--854,
  2013.

\bibitem{Ore1961}
O.~Ore, {\em Theory of graphs}, vol.~38.
\newblock Providence: Amer. Math. Soc. Transl.,, 1962.

\end{thebibliography}

\end{document}